\documentclass[12pt]{article}
\usepackage{amssymb,amsthm,amsmath,latexsym}
\usepackage{url}
\usepackage{lscape}
\newtheorem{thm}{Theorem}[section]
\newtheorem{prop}[thm]{Proposition}
\newtheorem{lem}[thm]{Lemma}
\newtheorem{cor}[thm]{Corollary}
\newtheorem{Q}{Question}
\theoremstyle{remark}
\newtheorem{rem}[thm]{Remark}

\newcommand{\ZZ}{\mathbb{Z}}
\newcommand{\RR}{\mathbb{R}}

\DeclareMathOperator{\wt}{wt}

\begin{document}
\title{Construction of extremal Type~II $\ZZ_{2k}$-codes}

\author{
Masaaki Harada\thanks{
Research Center for Pure and Applied Mathematics,
Graduate School of Information Sciences,
Tohoku University, Sendai 980--8579, Japan.
email: {\tt mharada@tohoku.ac.jp}.}
}


\maketitle

\begin{abstract}
We give methods for constructing many self-dual $\ZZ_m$-codes
and Type~II $\ZZ_{2k}$-codes of length $2n$
starting from a given self-dual $\ZZ_m$-code and
Type~II $\ZZ_{2k}$-code of length $2n$, respectively.
As an application, 
we construct extremal Type~II $\ZZ_{2k}$-codes
of length $24$ for $k=4,5,\ldots,20$ and 
extremal Type~II $\ZZ_{2k}$-codes of length $32$ for $k=4,5,\ldots,10$.
We also construct new extremal Type~II $\ZZ_4$-codes of lengths $56$ and $64$.
\end{abstract}

\section{Introduction}
A $\ZZ_{m}$-code of length $n$
is a $\ZZ_{m}$-submodule of $\ZZ_{m}^n$,
where $m$ is a positive integer with $m \ge 2$ and 
$\ZZ_{m}$ denotes the ring of integers modulo $m$.
A $\ZZ_{m}$-code $C$ of length $n$ is {\em self-dual} if $C=C^\perp$, where
$C^\perp$ denotes the dual code of $C$.
Self-dual codes are one of the most interesting classes of codes.
This interest is justified by many combinatorial objects
and algebraic objects related to self-dual codes
(see e.g.\ \cite{RS-Handbook} and the references given therein).

Many methods for constructing self-dual codes are known.
For example, 
starting from a given self-dual code with generator matrix 
$\left( \begin{array}{cc}
I_n & A
\end{array}\right)$,
by transforming $A$,
some methods for constructing many self-dual codes are known,
where $I_n$ denotes the identity matrix of order $n$
(see e.g.\ \cite{H98}, \cite{HK}, \cite{IS} and~\cite{Tonchev}).
In this paper,
we give a new method for constructing many self-dual $\ZZ_m$-codes
starting from a given self-dual $\ZZ_m$-code
with generator matrix
$\left( \begin{array}{cc}
I_n & A
\end{array}\right)$
by transforming $A$.
For self-dual codes over finite fields,
the corresponding result can be found in~\cite{IS}.

Now let us consider the case $m=2k$.
A binary doubly even self-dual code is often called Type~II.
Self-dual $\ZZ_{2k}$-codes with the property that all
Euclidean weights are multiples of $4k$ are called
\emph{Type~II} $\ZZ_{2k}$-codes
(see~\cite{Z4-BSBM} and~\cite{Z4-HSG} for $k=2$ and~\cite{BDHO} for $k \ge 3$).
Type~II $\ZZ_{2k}$-codes are a remarkable class of self-dual codes
related to even unimodular lattices.
For example, by Construction A,
Type~II $\ZZ_{2k}$-codes give even unimodular lattices.
There is a Type~II $\ZZ_{2k}$-code of length $n$
if and only if $n$ is divisible by eight~\cite{BDHO}.
For Type~II $\ZZ_{2k}$-codes $C$ of length $n$, 
the following upper bound on the minimum Euclidean weight
\begin{equation}\label{eq:B}
d_E(C) \le  4k  \left\lfloor \frac{n}{24} \right\rfloor +4k
\end{equation}
holds if $k=1$~\cite{MS}, if $k=2$~\cite{Z4-BSBM} (see also~\cite{Z4-HSG})
and if
$k\in \{3,4,5,6\}$~\cite{HM}.
Also, the bound~\eqref{eq:B} holds under the
assumption that $\lfloor n/24 \rfloor\leq k-2$ for arbitrary $k$~\cite{BDHO}.
When the bound~\eqref{eq:B} holds, 
we say that a Type~II $\ZZ_{2k}$-code meeting~\eqref{eq:B}
with equality is \emph{extremal}.
There have been significant researches on constructing
extremal Type~II $\ZZ_{2k}$-codes for $k=1$ and $2$
(see e.g.\ \cite{Z4-BSBM}, \cite{H98}, \cite{HK}, \cite{Z4-HSG},
\cite{RS-Handbook}, \cite{Tonchev} and the references given therein).
In this paper, we also give a new method for constructing many
Type~II $\ZZ_{2k}$-codes
starting from a given Type~II $\ZZ_{2k}$-code
with generator matrix
$\left( \begin{array}{cc}
I_n & A
\end{array}\right)$
by transforming $A$.
This generalizes the methods given in~\cite{H98} and~\cite{IS}.
By four-negacirculant codes and by the new method,
we construct
extremal Type~II $\ZZ_{2k}$-codes of length $24$ for $k=4,5,\ldots,20$ 
and 
extremal Type~II $\ZZ_{2k}$-codes
of length $32$ for $k=4,5,\ldots,10$. 
We also construct new extremal Type~II $\ZZ_4$-codes of lengths $56$ and $64$.

The paper is organized as follows.
In Section~\ref{sec:2},
we give the definitions and basic facts
on self-dual $\ZZ_m$-codes and Type~II $\ZZ_{2k}$-codes
used throughout this paper.
In Section~\ref{sec:method},
we give a method for constructing self-dual $\ZZ_m$-codes
(resp.\ Type~II $\ZZ_{2k}$-codes) of length $2n$
starting from a given self-dual $\ZZ_m$-code (resp.\ Type~II $\ZZ_{2k}$-code)
of length $2n$
having generator matrix 
$\left( \begin{array}{cc}
I_n & A
\end{array}\right)$,
by transforming $A$ (Theorem~\ref{thm:main} (resp.\
Theorem~\ref{thm:II})).
In Section~\ref{sec:24}, we construct extremal Type~II $\ZZ_{2k}$-codes
of length $24$,
by four-negacirculant codes and by Theorem~\ref{thm:II}
for $k=4,5,\ldots,20$.
Consequently, it is shown that if $k \in \{2,3,\ldots,20\}$ 
then there is an extremal Type~II $\ZZ_{2k}$-code $C$
such that $C^{(2)} \cong B$ 
for every binary doubly even self-dual code $B$ of length $24$,
where $C^{(2)}$ denotes the binary part of $C$ and 
$C \cong D$ means that $C$ and $D$ are equivalent.
In Section~\ref{sec:32}, we construct extremal Type~II $\ZZ_{2k}$-codes
of length $32$,
by four-negacirculant codes and by Theorem~\ref{thm:II}
for $k=4,5,\ldots,10$.
Consequently, it is shown that if $k \in \{2,3,\ldots,10\}$
then there is an extremal Type~II $\ZZ_{2k}$-code $C$
such that $C^{(2)} \cong B$ for every binary extremal doubly even self-dual code 
$B$ of length $32$.
In Section~\ref{sec:56-64}, by four-negacirculant codes and by Theorem~\ref{thm:II},
we construct new extremal Type~II $\ZZ_4$-codes of lengths $56$ and $64$.

All computer calculations in this paper
were done with the help of {\sc Magma}~\cite{Magma}.

\section{Preliminaries}\label{sec:2}

In this section, we give the definitions and basic facts
on self-dual $\ZZ_m$-codes and Type~II $\ZZ_{2k}$-codes
used throughout this paper.


Let $\ZZ_{m}$ be the ring of integers modulo $m$, where $m$ is a positive
integer with $m \ge 2$.
A $\ZZ_{m}$-code $C$ of length $n$
is a $\ZZ_{m}$-submodule of $\ZZ_{m}^n$.
An element of $C$ is called a {\em codeword}.
A {\em generator matrix} of $C$ is a matrix whose rows generate $C$.
We define the inner product of
$x=(x_1,x_2,\ldots,x_n)$ and $y=(y_1,y_2,\ldots,y_n) \in \ZZ_m^n$ by
\[
\langle x,y\rangle  = x_1 y_1 + x_2y_2 + \cdots + x_n y_n.
\]
The {\em dual code} $C^\perp$ of $C$ is defined as
\[
C^\perp = \{ x \in \ZZ_{m}^n \mid \langle x,y\rangle = 0 \text{ for all }
y \in C\}.
\]
A code $C$ is {\em self-dual} if $C=C^\perp$.
Two $\ZZ_m$-codes $C$ and $C'$ are {\em equivalent},
denoted $C \cong C'$,
if one can be obtained from the
other by permuting the coordinates and (if necessary)
changing the signs of certain coordinates.

We now consider self-dual $\ZZ_{2k}$-codes.
In this paper, we take the set $\ZZ_{2k}$ to be either
$\{0,1,\ldots,2k-1\}$ or
$\{0,\pm 1,\ldots,\pm(k-1),k\}$, using whichever form is more convenient.
We define a map $\rho$ from $\ZZ_{2k}$ to $\ZZ$ as follows
\begin{equation}\label{eq:rho}
\rho(i)=
\begin{cases}
0, \ldots , k & \text{ if } i=0, \ldots , k, \text{ respectively,}\\
1-k,\ldots , -1 & \text{ if } i=k+1, \ldots , 2k-1, \text{ respectively.}
\end{cases}
\end{equation}
The {\em Euclidean weight} $\wt_E(v)$ of a vector $v=(v_1,v_2,\ldots,v_n) \in \ZZ_{2k}^n$
is
\[
\sum_{i=1}^{k} i^2 n_i(v)+
\sum_{i=k+1}^{2k-1} (2k-i)^2 n_i(v),
\]
where
\[
n_i(v)=|\{j \in \{1,2,\ldots,n\}\mid
\rho(v_j) \equiv i \pmod{2k} \}|
\ (i=0,1,\ldots,2k-1).
\]
A binary doubly even self-dual code is often called Type~II.
Using Euclidean weights, 
the notion of binary Type~II codes has been generalized as follows.
Self-dual $\ZZ_{2k}$-codes with the property that all
Euclidean weights are multiples of $4k$ are called
\emph{Type~II} $\ZZ_{2k}$-codes
(see~\cite{Z4-BSBM} and~\cite{Z4-HSG} for $k=2$ and~\cite{BDHO} for $k \ge 3$).
It was shown in~\cite{BDHO} that there is a Type~II $\ZZ_{2k}$-code of
length $n$ if and only if $n$ is divisible by eight.

Throughout this paper,
$I_n$ denotes the identity matrix of order $n$, and
$A^T$ denotes the transpose of a matrix $A$.
The following gives a criteria for self-dual $\ZZ_{2k}$-codes
and Type~II $\ZZ_{2k}$-codes.

\begin{lem}\label{lem:SD}
\begin{itemize}
\item[\rm (i)]
Let  $C$ be a $\ZZ_{2k}$-code of length $2n$ having
generator matrix of form
$\left( \begin{array}{cc}
I_n & A
\end{array}\right)$, 
where $A$ is an $n \times n$ matrix.
If $AA^T=-I_n$, then $C$ is self-dual.
\item[\rm (ii)]
Let  $C$ be a self-dual $\ZZ_{2k}$-code  having
generator matrix $G$.
If each row of $G$ has Euclidean weight divisible by $4k$, then
$C$ is Type~II.
\end{itemize}
\end{lem}
\begin{proof}
(i) is trivial.
(ii) follows from~\cite[Lemma~2.2]{BDHO}.
\end{proof}

The {\em minimum Euclidean weight} $d_E(C)$ of a Type~II $\ZZ_{2k}$-code $C$
is the smallest Euclidean
weight among all nonzero codewords of $C$.
For Type~II $\ZZ_{2k}$-codes of length $n$, 
the upper bound~\eqref{eq:B} on the minimum Euclidean weight $d_E(C)$
holds if $k=1$~\cite{MS}, if $k=2$~\cite{Z4-BSBM} and~\cite{Z4-HSG}
and if
$k\in \{3,4,5,6\}$~\cite{HM}.
Also, the bound~\eqref{eq:B} holds under the
assumption that $\lfloor n/24 \rfloor\leq k-2$ for arbitrary $k$~\cite{BDHO}.
When the bound~\eqref{eq:B} holds, 
we say that a Type~II $\ZZ_{2k}$-code meeting~\eqref{eq:B}
with equality is \emph{extremal}.

The {\em binary part} $C^{(2)}$ of a Type~II $\ZZ_{2k}$-code $C$
of length $n$ is defined as
\[
\{(\rho(c_1) \!\!\!\pmod{2},\rho(c_2) \!\!\!\pmod{2},\ldots,
\rho(c_n)\!\!\!\pmod{2}) 
\mid (c_1,c_2,\ldots,c_n) \in C\},
\]
where $\rho$ is the map given in~\eqref{eq:rho}.
If a Type~II $\ZZ_{2k}$-code 
has generator matrix of  form
$\left( \begin{array}{cc}
I_n & A
\end{array}\right)$, then the binary part 
is a binary doubly even self-dual code~\cite{DGH} and~\cite{DHS}.
The following is trivial.

\begin{lem}\label{lem:equiv}
Let  $C$ and $C'$ be Type~II $\ZZ_{2k}$-codes.
If $C \cong C'$, then
$C^{(2)} \cong {C'}^{(2)}$.
\end{lem}


An $n \times n$
\emph{negacirculant} matrix has the following form
\[
\left( \begin{array}{ccccc}
r_0&r_1&r_2& \cdots &r_{n-1} \\
-r_{n-1}&r_0&r_1& \cdots &r_{n-2} \\
-r_{n-2}&-r_{n-1}&r_0& \cdots &r_{n-3} \\
\vdots &\vdots & \vdots && \vdots\\
-r_1&-r_2&-r_3& \cdots&r_0
\end{array}
\right).
\]
Let $C$ be the $\ZZ_m$-code of length $4n$
having the following generator matrix
\begin{equation} \label{eq:4}
\left(
\begin{array}{ccc@{}c}
\quad & {\Large I_{2n}} & \quad &
\begin{array}{cc}
A & B \\
-B^T & A^T
\end{array}
\end{array}
\right),
\end{equation}
where $A$ and $B$ are $n \times n$ negacirculant matrices.
Such a code is called a \emph{four-negacirculant} code~\cite{HHKK}.
By Lemma~\ref{lem:SD} (i), 
if  $AA^T+BB^T=-I_n$, then $C$ is self-dual.
Let $C$ be a four-negacirculant self-dual $\ZZ_{2k}$-code with
generator matrix of form~\eqref{eq:4}.
Let $r_A$ and $r_B$ denote the first rows of $A$ and $B$, respectively.
By Lemma~\ref{lem:SD} (ii), 
if $\wt_E(r_A)+\wt_E(r_B) \equiv -1 \pmod{4k}$, then
$C$ is Type~II.

We end this section by describing unimodular lattices and Construction A.
A (Euclidean) integral lattice $L \subset \RR^n$
in dimension $n$
is {\em unimodular} if $L = L^{*}$, where
$L^{*}$ is the dual lattice under the standard inner product $(x,y)$.
A unimodular lattice $L$ is {\em even}
if the norm $(x,x)$ of every vector $x$ of $L$ is even,
and {\em odd} otherwise.
There is an even unimodular lattice in dimension $n$
if and only if $n$ is divisible by eight.
The {\em minimum norm} $\min(L)$ of
a unimodular lattice $L$ is the smallest
norm among all nonzero vectors of $L$.
The {\em kissing number} of $L$ is the number of vectors of minimum norm in $L$.
We now give a method to construct
even unimodular lattices from Type~II $\ZZ_{2k}$-codes, which
is called Construction A.
%
For a Type~II $\ZZ_{2k}$-code $C$ of length $n$, 
define the following lattice
\begin{equation}\label{eq:A}
A_{2k}(C)=\frac{1}{\sqrt{2k}}\{\rho (C) +2k \ZZ^{n}\},
\end{equation}
where
\[
\rho (C)=\{(\rho (c_{1}), \rho (c_{2}), \ldots , \rho (c_{n}))\
\vert\ (c_{1}, c_2,\ldots , c_{n}) \in C\}.
\]
It is known that $A_{2k}(C)$ is an even unimodular lattice
with minimum norm $\min\{d_E(C)/2k,2k\}$,
where $d_E(C)$ denotes the minimum Euclidean weight of $C$~\cite{BDHO}.
The minimum norms of $A_{2k}(C)$ are used to determine the
minimum Euclidean weights of $C$ in Sections~\ref{sec:24},
\ref{sec:32} and~\ref{sec:56-64}.

\section{Methods for constructing self-dual $\ZZ_m$-codes}
\label{sec:method}

In this section, starting from a given self-dual $\ZZ_m$-code
(resp.\ Type~II $\ZZ_{2k}$-code)
with generator matrix
$\left( \begin{array}{cc}
I_n & A
\end{array}\right)$, by transforming $A$,
we give a new method for constructing many
self-dual $\ZZ_m$-codes
(resp.\ Type~II $\ZZ_{2k}$-codes).

\begin{thm}\label{thm:main}
  Let $C$ be a self-dual $\ZZ_m$-code of length $2n$
  having generator matrix 
$\left( \begin{array}{cc}
I_n & A
\end{array}\right)$.
Let $r_i$ be the $i$-th row of $A$.
Let $x$ and $y$ be vectors of $\ZZ_m^{n}$.
Define an $n \times n$ 
matrix $A(x,y)$, where the $i$-th row $r'_i$ is given by
\[
r'_i = r_i + \langle r_i,y \rangle x- \langle r_i,x \rangle y.
\]
Let $C(A,x,y)$ be the $\ZZ_m$-code having the following generator matrix
\[
\left( \begin{array}{cc}
I_n & A(x,y)
\end{array}\right).
\]
If
$
\langle x,x \rangle=
\langle y,y \rangle=
\langle x,y \rangle=0$,
then $C(A,x,y)$ is a self-dual $\ZZ_m$-code of length $2n$.
\end{thm}
\begin{proof}
Since 
$
\langle x,x \rangle=
\langle y,y \rangle=
\langle x,y \rangle=0$,
we have
\begin{equation}\label{eq:inner}
\begin{split}
\langle r'_i,r'_j \rangle
=&
\langle r_i,r_j \rangle
+ \langle r_j,y \rangle \langle r_i,x \rangle
- \langle r_j,x \rangle \langle r_i,y \rangle
\\&
+ \langle r_i,y \rangle \langle x,r_j \rangle
+ \langle r_i,y \rangle \langle r_j,y \rangle \langle x,x \rangle
- \langle r_i,y \rangle \langle r_j,x \rangle \langle x,y \rangle
\\&
- \langle r_i,x \rangle \langle y,r_j \rangle
- \langle r_i,x \rangle \langle r_j,y \rangle \langle y,x \rangle
+ \langle r_i,x \rangle \langle r_j,x \rangle \langle y,y \rangle
\\
=&
\langle r_i,r_j \rangle
+ \langle r_j,y \rangle \langle r_i,x \rangle
- \langle r_j,x \rangle \langle r_i,y \rangle
+ \langle r_i,y \rangle \langle x,r_j \rangle
\\&
- \langle r_i,x \rangle \langle y,r_j \rangle
\\
=&\langle r_i,r_j \rangle.
\end{split}
\end{equation}
This implies that $AA^T=A(x,y)A(x,y)^T$.
By Lemma~\ref{lem:SD} (i), $C(A,x,y)$ is self-dual.
\end{proof}

\begin{rem}
For self-dual codes over finite fields, the corresponding result can be found in~\cite{IS}.
\end{rem}

\begin{rem}
We modify the above theorem as follows.
Let $C$ be a $\ZZ_m$-code of length $n+n'$ having generator matrix 
$\left( \begin{array}{cc}
I_n & A
\end{array}\right)$,
where $A$ is an $n \times n'$ matrix.
Then the code $C(A,x,y)$ having generator matrix
$\left( \begin{array}{cc}
I_n & A(x,y)
\end{array}\right)$
defined in Theorem~\ref{thm:main} satisfies 
$AA^T=A(x,y)A(x,y)^T$.
\end{rem}

We consider a method for constructing Type~II $\ZZ_{2k}$-codes.

\begin{thm}\label{thm:II}
Suppose that $C$ is a Type~II $\ZZ_{2k}$-code of length $2n$ having
generator matrix 
$\left( \begin{array}{cc}
I_n & A
\end{array}\right)$.
Suppose that
$x$ and $y$ are vectors of $\ZZ_{2k}^{n}$ satisfying the condition
$\wt_E(x) \equiv \wt_E(y) \equiv 0 \pmod{4k}$ and 
$\langle x,y \rangle=0$.
Then the code $C(A,x,y)$ having generator matrix
$\left( \begin{array}{cc}
I_n & A(x,y)
\end{array}\right)$
defined in Theorem~\ref{thm:main}
is a Type~II $\ZZ_{2k}$-code of length $2n$.
\end{thm}
\begin{proof}
Let $\rho$ be the map given in~\eqref{eq:rho}.
For a vector $v=(v_1,v_2,\ldots,v_n) \in \ZZ_{2k}$,
define the positive integer
$v*v=\sum_{i=1}^{2k-1} i^2 n_i(v)$, where
$n_i(v)=|\{j \in \{1,2,\ldots,n\}\mid
\rho(v_j) \equiv i \pmod{2k} \}|$ $(i=0,1,\ldots,2k-1)$.
By~\cite[Lemma~2.1]{BDHO},
$\wt_E(v) \equiv v*v \pmod{4k}$.
This implies that $\langle x,x \rangle=\langle y,y \rangle=0$.
Thus, by Theorem~\ref{thm:main}, $C(A,x,y)$ is self-dual.
Let $r_i$ and $r'_i$ denote the $i$-th row of $A$ and $A(x,y)$, respectively.
By Lemma~\ref{lem:SD} (ii), 
it is sufficient to show that $\wt_E(r'_i) \equiv 4k-1 
\pmod{4k}$ $(i=1,2,\ldots,n)$.
From~\eqref{eq:inner}, we have
\begin{align*}
\wt_E(r'_i)&\equiv
r'_i * r'_i 
\\&\equiv
r_i * r_i 
+ x*x+ y*y
\\&\equiv
\wt_E(r_i)
+\wt_E(x)
+\wt_E(y) 
\\&
\equiv
\wt_E(r_i) 
\\&\equiv 4k-1 \pmod{4k}.
\end{align*}
The result follows.
\end{proof}

\begin{rem}
For the case $k=1$, the above theorem can be found in~\cite{IS}.
Theorem~5 in~\cite{H98} can be obtained from the above theorem
by putting $k=2$, $x=t$ and $y=(2,2,\ldots,2)$.
\end{rem}

\section{Extremal Type~II $\ZZ_{2k}$-codes of length 24}
\label{sec:24}

For $k=4,5,\ldots,20$,
in this section, we construct extremal Type~II $\ZZ_{2k}$-codes
of length $24$,
by four-negacirculant codes and by Theorem~\ref{thm:II}.

\subsection{Motivation and results}
For lengths $8$ and $16$, 
every Type~II $\ZZ_{2k}$-code is extremal.
In other words, 
the smallest possible length for which there is a nontrivial extremal Type~II $\ZZ_{2k}$-code is $24$.
Meanwhile, an extremal Type~II $\ZZ_{2k}$-codes of length $24$ are related to the Leech lattice, which is arguably the most remarkable lattice.
In particular, by considering the existence of $2k$-frames in the Leech
lattice, it was shown that there is an extremal Type~II $\ZZ_{2k}$-code
of length $24$ for every positive integer $k$ with
$k \ge 2$~\cite{Chapman} and~\cite{GHZ22}.

There are nine inequivalent binary doubly even self-dual codes
of length $24$~\cite{PS75}.
The seven codes are indecomposable and these code
are denoted by $\textrm{A}24, \textrm{B}24, \ldots, \textrm{G}24$
in~\cite[Table~II]{PS75}.
Only the code $\textrm{G}24$ has minimum weight $8$, that is, 
$\textrm{G}24$ is extremal, and it is well known as
the extended Golay code.
As usual we denote the two decomposable codes by $e_8^3$ and $d_{16}e_8$.
For every binary doubly even self-dual code $B$ of length $24$, 
there is an extremal Type~II $\ZZ_{2k}$-code $C$ 
such that $C^{(2)} \cong B$ when 
$k=2$~\cite[Proposition~9]{H98} (see~\cite[Postscript]{CS97})
and $k=3$~\cite[Theorem~3]{HK02}.
In this section, we establish the following theorem
by constructing extremal Type~II $\ZZ_{2k}$-codes of length $24$
for $k=4,5,\ldots,20$
by using four-negacirculant codes and by Theorem~\ref{thm:II} explicitly.

\begin{thm}\label{thm:24}
Suppose that $k \in \{2,3,\ldots,20\}$.
For every binary doubly even self-dual code $B$ of length $24$,
there is an extremal Type~II $\ZZ_{2k}$-code $C$
such that $C^{(2)} \cong B$.
\end{thm}

Although the following lemma is trivial, we give a proof for completeness.

\begin{lem}\label{lem:min}
Suppose that $k \ge 3$ and $n \in\{24,32,40\}$.
Let $C$ be a Type~II $\ZZ_{2k}$-code of length $n$.
Let $A_{2k}(C)$ denote the even unimodular lattice given in~\eqref{eq:A}.
Then $C$ is extremal if and only if $A_{2k}(C)$ has minimum norm $4$.
\end{lem}
\begin{proof}
Since $k \ge 3$,
a vector of norm $2$ of $A_{2k}(C)$ is written as
\[
\frac{1}{\sqrt{2k}}
(\rho (c_{1}), \rho (c_{2}), \ldots , \rho (c_{n}))
\]
for some codeword $c=(c_{1}, c_2,\ldots , c_{n})  \in C$ such that
$\wt_E(c)=4k$,
where $\rho$ is the map given in~\eqref{eq:rho}.
Thus, 
$C$ contains no codeword of Euclidean weight $4k$
if and only if $A_{2k}(C)$ contains no vector of norm $2$.
It is known that the minimum norm $\min(L)$ of an even unimodular lattice
in dimension $n$
is bounded by $\min(L) \le 2 \lfloor \frac{n}{24} \rfloor +2$
(see~\cite[Chapter~7]{SPLAG}).
Thus, $\min(A_{2k}(C)) \le 4$ for $n \in\{24,32,40\}$.
Hence, 
$C$ contains no codeword of Euclidean weight $4k$
if and only if $A_{2k}(C)$ has minimum norm $4$.
The result follows.
\end{proof}

By the above lemma, 
we determined the extremality for all Type~II $\ZZ_{2k}$-codes
found in this section and
the next section.

\subsection{Extremal Type~II $\ZZ_8$-codes of length 24}

By considering four-negacirculant codes, our computer search found
extremal Type~II $\ZZ_8$-codes $C_{8,24,1}$ and $C_{8,24,2}$ of length $24$
(see the end of Section~\ref{sec:2} for four-negacirculant codes).
The codes $C_{8,24,1}$ and $C_{8,24,2}$ have generator matrices
of form~\eqref{eq:4},
where the first rows $r_A$ and $r_B$ of the negacirculant matrices $A$ and $B$
are as follows
\begin{align*}
  (r_A, r_B)=&((5,3,3,3,3,1), (5,5,5,6,7,7)) \text{ and }\\
  &((0,0,0,1,1,3), (0,1,0,3,3,1)),  
\end{align*}
respectively.
By Theorem~\ref{thm:II}, our computer search found more extremal Type~II
$\ZZ_{8}$-codes $C_{8,24,3},C_{8,24,4},\ldots,C_{8,24,8}$  and
$C_{8,24,9}$ as $C(A,x,y)$ from 
$M_{8,24,i}$ $(i=1,2,4)$, where $M_{8,24,i}$
denotes the right half of the generator matrix of
$C_{8,24,i}$.
In Table~\ref{Tab:Z8}, we list 
the matrices $A$ and the vectors $x,y$ in Theorem~\ref{thm:II}.

\begin{table}[thb]
\caption{Extremal Type~II $\ZZ_8$-codes of length $24$}
\label{Tab:Z8}
\centering
\medskip
{\footnotesize
\begin{tabular}{c|c|cc}
\noalign{\hrule height0.8pt}
Codes  & $A$ & $x$&$y$ \\
\hline
$C_{8,24,3}$&$M_{8,24,1}$&$(1,1,1,4,7,3,3,0,2,5,5,6)$&$(0,0,0,1,1,2,4,6,0,7,2,1)$\\
$C_{8,24,4}$&$M_{8,24,1}$&$(5,4,1,4,1,1,3,1,3,0,7,4)$&$(0,1,1,1,7,4,6,6,1,3,3,7)$\\
$C_{8,24,5}$&$M_{8,24,1}$&$(6,5,1,4,5,5,6,6,4,6,0,2)$&$(0,0,0,5,2,2,2,5,6,2,3,1)$\\
$C_{8,24,6}$&$M_{8,24,2}$&$(3,3,3,6,1,0,3,5,1,2,5,4)$&$(0,0,0,3,1,6,4,5,2,1,4,2)$\\
$C_{8,24,7}$&$M_{8,24,2}$&$(4,3,7,0,0,0,6,2,6,0,1,5)$&$(0,1,6,1,2,4,5,6,0,0,0,5)$\\
$C_{8,24,8}$&$M_{8,24,2}$&$(4,1,1,7,7,5,2,3,0,3,5,6)$&$(0,0,2,3,0,2,2,0,3,3,0,3)$\\
$C_{8,24,9}$&$M_{8,24,4}$&$(5,2,2,6,6,2,6,1,5,2,1,0)$&$(0,0,1,6,5,0,6,0,1,0,3,2)$\\
\noalign{\hrule height0.8pt}
\end{tabular}
}
\end{table}

\subsection{Extremal Type~II $\ZZ_{2k}$-codes of length 24
for $k=5,6,\ldots,20$}

For $\ZZ_{2k}$-codes $(k=5,6,\ldots,20)$, 
to save space, we list the results only.
Note that the approach is similar to that in the previous subsection.

For $k =5,6,\ldots,20$,
by considering four-negacirculant codes, our computer search found
extremal Type~II $\ZZ_{2k}$-codes 
$C_{2k,24,i}$ of length $24$ $(i=1,2,\ldots,i_k)$,
where 
$i_k=3$ if $k \in \{5,6,\ldots,10\}$ and
$i_k=4$ if $k \in \{11,12,\ldots,20\}$.
The codes  have generator matrices
of form~\eqref{eq:4},
where the first rows $r_A$ and $r_B$ of the negacirculant matrices $A$ and $B$
are listed in Table~\ref{Tab:24-2}.
By Theorem~\ref{thm:II}, our computer search found more extremal Type~II
$\ZZ_{2k}$-codes $C_{2k,24,j}$  $(j=i_k+1,i_k+2,\ldots,9)$ 
as $C(A,x,y)$ from $M_{2k,24,i}$, where $M_{2k,24,i}$
denotes the right half of the generator matrix of
$C_{2k,24,i}$
$(i=1,2,\ldots,9)$.
In Tables~\ref{Tab:24-xy1}, \ref{Tab:24-xy2} and~\ref{Tab:24-xy3}, we list 
the matrices $A$ and the vectors $x,y$ in Theorem~\ref{thm:II}.

\subsection{Binary parts of $C_{2k,24,1},C_{2k,24,2}, \ldots, C_{2k,24,9}$
for $k=4,5,\ldots,20$}

For a given binary doubly even self-dual code $B$ of length $24$
and a given $k$, 
we list in Table~\ref{Tab:Fig} the number $i$ such that
$C^{(2)}_{2k,24,i} \cong B$. 
From Table~\ref{Tab:Fig}, 
we have Theorem~\ref{thm:24}
combining with known results in~\cite{H98} and~\cite{HK02}.

\begin{rem}
\begin{itemize}
\item[(i)]
For $B= e_8^3 $ and $\textrm{G}24$,
an extremal Type~II $\ZZ_8$-code $C$ of length $24$
such that $C^{(2)} \cong B$
can be found in~\cite{GHK} and~\cite{GH}, respectively.
\item[(ii)]
Extremal Type~II $\ZZ_{10}$-codes $C$ of length $24$ such that
$C^{(2)} \cong e_8^3$
can be found in~\cite{GHK}.
\item[(iii)]
For $B= \textrm{E}24$ and $e_8^3$, 
an extremal Type~II $\ZZ_{12}$-code $C$ of length $24$ such that
$C^{(2)} \cong B$
can be found in~\cite{GHK}.
\item[(iv)]
An extremal Type~II $\ZZ_{14}$-code $C$ of length $24$ such that
$C^{(2)} \cong e_8^3$ can be found in~\cite{GHK}.
\item[(v)]
An extremal Type~II $\ZZ_{16}$-code $C$ of length $24$ such that
$C^{(2)} \cong \textrm{E}24$ can be found in~\cite{GHK}.
\item[(vi)]
An extremal Type~II $\ZZ_{22}$-code $C$ of length $24$ such that
$C^{(2)} \cong \textrm{E}24$ can be found in~\cite{GHZ22}.
For $B= \textrm{E}24$ and $\textrm{G}24$,
an extremal Type~II $\ZZ_{22}$-code $C$ of length $24$ such that
$C^{(2)} \cong B$ can be found in~\cite{GHK}.
\end{itemize}
\end{rem}

In addition, by Lemma~\ref{lem:equiv}, we have the following corollary.

\begin{cor}
Suppose that $k \in \{2,3,\ldots,20\}$.
Then there are at least nine inequivalent extremal Type~II $\ZZ_{2k}$-codes of length $24$.
\end{cor}

\begin{table}[thb]
\caption{Binary parts of $C_{2k,24,i}$ $(k=4,5,\ldots,20)$}
\label{Tab:Fig}
\centering
\medskip
{\small
\begin{tabular}{c|ccccccccc}
\noalign{\hrule height0.8pt}
$k$ & $\textrm{A}24$& $\textrm{B}24$& $\textrm{C}24$& $\textrm{D}24$&
$\textrm{E}24$& $\textrm{F}24$& $\textrm{G}24$& $e_8^3$& $d_{16}e_8$\\
\hline
 4 & 3& 8& 5& 6& 1& 2& 7& 4& 9 \\
 5 & 3& 6& 5& 4& 2& 9& 7& 1& 8 \\
 6 & 3& 5& 6& 4& 2& 7& 9& 1& 8 \\
 7 & 3& 4& 7& 6& 2& 8& 9& 1& 5 \\
 8 & 4& 7& 6& 5& 2& 3& 8& 1& 9 \\
 9 & 3& 5& 4& 8& 2& 7& 9& 1& 6 \\
10 & 3& 5& 7& 4& 2& 8& 9& 1& 6 \\
11 & 3& 5& 8& 7& 2& 4& 9& 1& 6 \\
12 & 2& 8& 6& 7& 3& 4& 9& 1& 5 \\
13 & 4& 5& 8& 6& 3& 2& 9& 1& 7 \\
14 & 2& 7& 9& 5& 4& 3& 8& 1& 6 \\
15 & 2& 5& 7& 6& 4& 3& 9& 1& 8 \\
16 & 2& 5& 7& 8& 4& 3& 9& 1& 6 \\
17 & 2& 7& 5& 8& 3& 4& 9& 1& 6 \\
18 & 2& 5& 8& 6& 3& 4& 9& 1& 7 \\
19 & 2& 6& 5& 8& 4& 3& 9& 1& 7 \\
20 & 2& 7& 5& 8& 4& 3& 9& 1& 6 \\
\noalign{\hrule height0.8pt}
\end{tabular}
}
\end{table}

We end this section with the following question.

\begin{Q}
Suppose that $k \ge 21$.
For every binary doubly even self-dual code $B$ of length $24$,
is there an extremal Type~II $\ZZ_{2k}$-code $C$
such that $C^{(2)} \cong B$.
\end{Q}

\section{Extremal Type~II $\ZZ_{2k}$-codes of length 32}
\label{sec:32}

For $k=4,5,\ldots,10$,
in this section, we construct extremal Type~II $\ZZ_{2k}$-codes
of length $32$,
by four-negacirculant codes and by Theorem~\ref{thm:II}.

There are five inequivalent binary extremal doubly even self-dual codes
of length $32$~\cite{CPS}.
The five codes are denoted by $\textrm{C}81, \textrm{C}82, \ldots, \textrm{C}85$
in~\cite[Table~A]{CPS}.
For every binary extremal doubly even self-dual code $B$ of length $32$, 
there is an extremal Type~II $\ZZ_{2k}$-code $C$ 
such that $C^{(2)} \cong B$ when 
$k=2$~\cite[Theorem~5]{GH32} and $k=3$~\cite[Proposition~3]{HZ6}.

In this section, 
by constructing extremal Type~II $\ZZ_{2k}$-codes of length $32$
($k=4,5,\ldots,10$)
by using four-negacirculant codes and by Theorem~\ref{thm:II},
we explicitly establish the following theorem.

\begin{thm}\label{thm:32}
Suppose that $k \in \{2,3,\ldots,10\}$.
For every binary extremal doubly even self-dual code $B$ of length $32$,
there is an extremal Type~II $\ZZ_{2k}$-code $C$
such that $C^{(2)} \cong B$.
\end{thm}

\begin{table}[thb]
\caption{Extremal Type~II  four-negacirculant $\ZZ_{2k}$-codes of length $32$ $(k=4,5,\ldots,10)$}
\label{Tab:32}
\centering
\medskip
{\small
\begin{tabular}{c|ll}
\noalign{\hrule height0.8pt}
Codes & \multicolumn{1}{c}{$r_A$}&\multicolumn{1}{c}{$r_B$} \\
\hline
$C_{8,32,1}$&$(7,4,3,6,3,5,6,3)$&$(5,6,0,2,0,4,7,4)$\\
$C_{8,32,2}$&$(7,4,3,6,4,5,2,5)$&$(6,5,2,2,0,7,1,2)$\\
$C_{10,32,1}$&$(5,1,4,5,1,8,0,4)$&$(4,6,0,0,2,3,1,5)$\\
$C_{10,32,2}$&$(5,1,4,5,4,0,5,0)$&$(6,6,5,2,8,4,9,7)$\\
$C_{12,32,1}$&$( 2,  5, 11,  3,  1,  7,  6, 10)$&$(  5, 10,  5,  6,  8,  6,  0,  4)$\\
$C_{12,32,2}$&$( 2,  5, 11,  3,  3,  4, 10,  2)$&$(  0,  6, 10,  1,  2,  0,  5,  1)$\\
$C_{14,32,1}$&$( 3, 11,  6,  1,  3, 10, 13, 12)$&$(  4, 10,  3,  0,  7, 12,  6,  8)$\\
$C_{14,32,2}$&$( 3, 11,  6,  1,  2,  2, 11,  8)$&$( 12,  0,  2, 13,  1, 10,  6,  9)$\\
$C_{16,32,1}$&$( 2, 13,  4,  2, 11,  2,  0, 15)$&$(  4, 15,  9,  0,  2,  0, 15, 11)$\\
$C_{16,32,2}$&$( 2, 13,  4,  2, 14,  7,  7,  0)$&$(  6, 12,  7,  7,  5,  1,  6, 10)$\\
$C_{18,32,1}$&
$(14, 17, 15, 11,  1,  3, 14,  2)$&$(12, 10,  3,  4,  5,  6,  4,  4)$\\
$C_{18,32,2}$&
$(14, 17, 15, 11,  5, 16, 16, 14)$&$( 8, 16,  5,  6,  0,  3, 15, 14)$\\
$C_{20,32,1}$&
$(2,0,8,9,15,5,11,7)$&$(10,0,15,8,5,4,8,16)$\\
$C_{20,32,2}$&
$(2,0,8,9,1,14,12,1)$&$(6,18,0,1,1,19,17,16)$\\
\noalign{\hrule height0.8pt}
\end{tabular}
}
\end{table}

The approach is similar to that in the previous section.
To save space, we list the results only.
By considering four-negacirculant codes, our computer search found
extremal Type~II $\ZZ_{2k}$-codes 
$C_{2k,32,1}$ and $C_{2k,32,2}$ of length $32$
$(k =4,5,\ldots,10)$.
The codes  have generator matrices
of form~\eqref{eq:4},
where the first rows $r_A$ and $r_B$ of the negacirculant matrices $A$ and $B$
are listed in Table~\ref{Tab:32}.
By Theorem~\ref{thm:II}, our computer search found more extremal Type~II
$\ZZ_{2k}$-codes $C_{2k,32,3}$, $C_{2k,32,4}$ and $C_{2k,32,5}$  
as $C(A,x,y)$ from $M_{2k,32,i}$, where $M_{2k,32,i}$
denotes the right half of the generator matrix of
$C_{2k,32,i}$
$(i=1,2, k =4,5,\ldots,10)$.
In Table~\ref{Tab:32-2}, we list 
the matrices $A$ and the vectors $x,y$ in Theorem~\ref{thm:II}.
For a given binary extremal doubly even self-dual code $B$ of length $32$
and a given $k$, 
we list in Table~\ref{Tab:Fig32} the number $i$ such that
$C^{(2)}_{2k,32,i} \cong B$ $(k =4,5,\ldots,10)$. 
From Table~\ref{Tab:Fig32}, 
we have Theorem~\ref{thm:32}
combining with known results in~\cite{GH32} and~\cite{HZ6}.
In addition, by Lemma~\ref{lem:equiv}, we have the following corollary.

\begin{cor}
Suppose that $k \in \{2,3,\ldots,10\}$.
Then there are at least five inequivalent extremal Type~II $\ZZ_{2k}$-codes of length $32$.
\end{cor}

\begin{table}[thb]
\caption{Binary parts of $C_{2k,32,i}$ $(k=4,5,\ldots,10)$}
\label{Tab:Fig32}
\centering
\medskip
{\small
\begin{tabular}{c|ccccc}
\noalign{\hrule height0.8pt}
$k$ & $\textrm{C}81$ &$\textrm{C}82$ &$\textrm{C}83$ &$\textrm{C}84$ &$\textrm{C}85$ \\
\hline
 4&2&4&1&3&5\\
 5&2&4&1&5&3\\
 6&2&5&1&4&3\\
 7&2&4&1&3&5\\
 8&2&4&1&5&3\\
 9&2&4&1&5&3\\
10&2&5&1&3&4\\
\noalign{\hrule height0.8pt}
\end{tabular}
}
\end{table}


We end this section with the following question.

\begin{Q}
Suppose that $k \ge 11$.
For every binary extremal doubly even self-dual code $B$ of length $32$,
is there an extremal Type~II $\ZZ_{2k}$-code $C$
such that $C^{(2)} \cong B$.
\end{Q}

\section{Extremal Type~II $\ZZ_{4}$-codes of lengths 56 and 64}
\label{sec:56-64}

In this section, by four-negacirculant codes and by Theorem~\ref{thm:II},
we construct new extremal Type~II $\ZZ_4$-codes of lengths $56$ and $64$.

\subsection{Extremality}
We describe how to determine the extremality of a given Type~II $\ZZ_4$-code
of lengths $48,56$ and $64$.
Although the following lemma is trivial, we give a proof for completeness.

\begin{lem}\label{lem:extremal}
Suppose that $n \in \{48,56,64\}$.
Let $C$ be a Type~II $\ZZ_4$-code of length $n$.
Let $A_4(C)$ denote the even unimodular lattice given in~\eqref{eq:A}.
Then $C$ is extremal if and only if $A_4(C)$ has minimum norm $4$ and kissing number $2n$.
\end{lem}
\begin{proof}
It is trivial that $C$ has a codeword of Euclidean weight $8$
if and only if $A_4(C)$ has a vector of norm $2$.
If $C$ has a codeword of Euclidean weight $16$,
then $A_4(C)$ has a vector of norm $4$.
In addition, 
$A_4(C)$ has $2n$ vectors of norm $4$, which have the following form
\[
(\pm 2,0,0,\ldots,0),
(0,\pm 2,0,\ldots,0), \ldots,
(0,0,\ldots,0,\pm 2).
\]
Hence, $A_4(C)$ has minimum norm $4$ and kissing number $2n$
if and only if $C$ has no codeword of Euclidean weight $16$.
\end{proof}

By the above lemma, 
we determined the extremality for all Type~II $\ZZ_4$-codes
found in this section.

\subsection{Extremal Type~II $\ZZ_{4}$-codes of length 64}

The first example of an extremal Type~II $\ZZ_4$-code of length $64$ was found 
in~\cite{H2010}.  Note that the binary part of the code is not self-dual.

By considering four-negacirculant codes, our computer search found four 
extremal Type~II $\ZZ_4$-codes $C_{4,64,i}$ $(i=1,2,3,4)$
of length $64$.
These codes have generator matrices
of form~\eqref{eq:4},
where the first rows $r_A$ and $r_B$ of the negacirculant matrices $A$ and $B$
are listed in Table~\ref{Tab:64}.

\begin{table}[thb]
\caption{Extremal Type~II four-negacirculant $\ZZ_{4}$-codes of length $64$}
\label{Tab:64}
\centering
\medskip
{\footnotesize
\begin{tabular}{c|ll}
\noalign{\hrule height0.8pt}
Codes & \multicolumn{1}{c}{$r_A$}&\multicolumn{1}{c}{$r_B$} \\
\hline
$C_{4,64,1}$&
$(1,0,1,1,3,3,2,3,3,2,2,3,3,2,3,2)$&$(2,0,1,1,3,1,2,3,3,2,3,1,2,0,3,0)$\\
$C_{4,64,2}$&
$(3,0,3,2,3,0,3,2,0,3,0,1,2,2,0,0)$&$(0,2,0,0,1,2,3,2,1,2,0,0,3,0,2,1)$\\
$C_{4,64,3}$&
$(3,0,0,0,2,0,1,0,2,3,2,2,3,0,1,1)$&$(1,2,1,3,1,3,2,0,0,3,3,2,1,0,1,2)$\\
$C_{4,64,4}$&
$(0,0,3,3,3,1,1,1,0,1,0,3,0,3,3,2)$&$(1,0,1,2,1,2,2,0,3,1,2,3,1,0,3,3)$\\
\noalign{\hrule height0.8pt}
\end{tabular}
}
\end{table}

In addition, by Theorem~\ref{thm:II}, our computer search found one more extremal Type~II
$\ZZ_4$-code $C_{4,64,5}$ as $C(M_{4,64,2},x,y)$, where
$M_{4,64,2}$ denotes the right half of the generator matrix of
$C_{4,64,2}$, and 
\begin{align*}
&x=( 1, 1, 3, 0, 1, 3, 1, 0, 3, 2, 3, 0, 2, 2, 1, 1, 1, 2, 1, 3, 2, 2, 1, 0,
0, 0, 0, 0, 0, 0, 1, 1), \\
&y=(1, 3, 2, 0, 0, 3, 2, 2, 0, 2, 0, 3, 1, 2, 1, 0, 2, 0, 2, 2, 0, 2, 2, 3,
1, 1, 1, 1, 2, 1, 2, 2).
\end{align*}

By Gleason's theorem (see~\cite{MS}), the weight enumerator of a
binary doubly even self-dual code of length $64$ and minimum weight
at least $8$ is written as
\[
W_{64}(a)=
1
+ a y^8
+( 2976  + 20 a) y^{12}
+( 454956  + 2 a) y^{16}
+ \cdots,
\]
where $a$ is the number of codewords of weight $8$.
Our computer search
verified that
the binary parts $C_{4,64,i}^{(2)}$ $(i=1,2,3,4,5)$ have
weight enumerators
$W_{64}(16)$,
$W_{64}(64)$,
$W_{64}( 0)$,
$W_{64}( 0)$,
$W_{64}(19)$, respectively.
In addition, our computer search verified that
the binary parts $C_{4,64,3}^{(2)}$ and $C_{4,64,4}^{(2)}$ are inequivalent.
By Lemma~\ref{lem:equiv}, we have the following proposition.

\begin{prop}
There are at least six inequivalent extremal Type~II $\ZZ_{4}$-codes of length $64$.
\end{prop}

\subsection{Extremal Type~II $\ZZ_{4}$-codes of length 56}
The first example of an extremal Type~II $\ZZ_4$-code of length $56$ was found 
in~\cite{H2010}.  Note that the binary part of the code is not self-dual.
Two more extremal Type~II $\ZZ_4$-codes of length $56$ were constructed 
by considering double circulant codes in~\cite{H2016}.
The two codes are denoted by 
$\mathcal{D}_{56,1}$ and $\mathcal{C}_{56}$ in~\cite{H2016}.
Note that the binary parts of these codes are binary extremal
doubly even self-dual codes.

By considering four-negacirculant codes, our computer search found
three extremal Type~II $\ZZ_4$-codes $C_{4,56,i}$ $(i=1,2,3)$
of length $56$.
These codes have generator matrices
of form~\eqref{eq:4},
where the first rows $r_A$ and $r_B$ of the negacirculant matrices $A$ and $B$
are listed in Table~\ref{Tab:56}.

By Gleason's theorem (see~\cite{MS}), the weight enumerator of a
binary doubly even self-dual code of length $56$ and minimum weight
at least $8$ is written as
\[
W_{56}(a)=
1
+ a y^8
+ (8190  + 6 a) y^{12}
+ (622314  - 83 a) y^{16} 
+ \cdots,
\]
where $a$ is the number of codewords of weight $8$.
Our computer search
verified that
the binary parts $C_{4,56,i}^{(2)}$ $(i=1,2,3)$ have
weight enumerators
$W_{56}(35)$,
$W_{56}(42)$,
$W_{56}(0)$, respectively.
In addition, our computer search verified that 
$\mathcal{D}_{56,1}^{(2)}$, $\mathcal{C}_{56}^{(2)}$
and $C_{4,56,3}^{(2)}$ are inequivalent.
By Lemma~\ref{lem:equiv}, we have the following proposition.

\begin{prop}
There are at least six inequivalent extremal Type~II $\ZZ_{4}$-codes of length $56$.
\end{prop}

\begin{table}[thb]
\caption{Extremal Type~II four-negacirculant $\ZZ_{4}$-codes of length $56$}
\label{Tab:56}
\centering
\medskip
{\footnotesize
\begin{tabular}{c|ll}
\noalign{\hrule height0.8pt}
Codes & \multicolumn{1}{c}{$r_A$}&\multicolumn{1}{c}{$r_B$} \\
\hline
$C_{4,56,1}$&
$(2,2,0,1,2,0,2,0,2,2,1,1,2,1)$&$(0,0,3,3,2,0,3,1,3,1,1,0,0,2)$\\
$C_{4,56,2}$&
$(2,2,2,0,1,0,1,0,1,2,3,0,3,1)$&$(2,3,0,0,0,1,3,2,1,2,0,2,2,1)$\\
$C_{4,56,3}$&
$(1,2,2,0,3,2,2,3,1,1,1,3,1,0)$&$(2,1,3,1,1,2,2,2,0,1,3,3,0,0)$\\
\noalign{\hrule height0.8pt}
\end{tabular}
}
\end{table}

By an approach is similar to that used in the previous subsection, 
we tried to construct a new extremal Type~II $\ZZ_{4}$-code of length $56$
by Theorem~\ref{thm:II}.
However, our extensive search failed to construct such a code.

\bigskip
\noindent
{\bf Acknowledgments.}
This work was supported by JSPS KAKENHI Grant Number 19H01802.
The author would like to thank the anonymous reviewers for useful comments.


\begin{landscape}

\begin{table}[thb]
\caption{Extremal Type~II four-negacirculant $\ZZ_{2k}$-codes of length $24$ 
$(k=5,12,\ldots,20)$}
\label{Tab:24-2}
\centering
\medskip
{\small
\begin{tabular}{c|ll||c|ll}
\noalign{\hrule height0.8pt}
Codes & \multicolumn{1}{c}{$r_A$}&\multicolumn{1}{c||}{$r_B$} &
Codes & \multicolumn{1}{c}{$r_A$}&\multicolumn{1}{c}{$r_B$} \\
\hline
$C_{10,24,1}$ & $(9,8,2,5,4,2)$&$(3,8,6,0,6,0)$ &
$C_{26,24,4}$&$(14, 13, 16,  2,  8, 20)$&$( 7,  7,  7, 11, 11,  9)$\\
$C_{10,24,2}$ & $(1,7,3,1,9,3)$&$(7,5,3,1,1,8)$ &
$C_{28,24,1}$&$(15,  8, 10,  7,  6,  8)$&$(22, 16,  9, 16, 22, 14)$\\
$C_{10,24,3}$ & $(5,8,1,6,5,0)$&$(5,7,2,5,2,1)$ &
$C_{28,24,2}$&$( 2, 23, 20, 27, 24,  1)$&$(22,  3,  0,  3, 21, 25)$\\
$C_{12,24,1}$ &$(6, 4, 8, 8, 5, 6)$&$(10, 9, 0, 2, 5,10)$ &
$C_{28,24,3}$&$(13, 15,  1, 18, 18, 20)$&$( 3,  7,  5, 22,  9, 22)$\\
$C_{12,24,2}$ &$(3, 7, 9, 7,11,11)$&$( 9, 9, 2, 7,11, 1)$ &
$C_{28,24,4}$&$(17,  5,  3, 13, 25, 14)$&$(21,  1,  5,  1, 27, 25)$\\
$C_{12,24,3}$ &$(6, 9, 8, 7, 4, 3)$&$( 1, 0, 1, 6, 9, 3)$ &
$C_{30,24,1}$&$( 1,  8, 28,  9, 12, 18)$&$( 6,  8,  4,  2,  1, 20)$\\
$C_{14,24,1}$ &$( 2,6,11,4, 0, 5)$&$( 6, 4,6,0,8,11)$ &
$C_{30,24,2}$&$(29, 28, 23, 13, 21, 14)$&$(22, 11,  4,  3,  0, 27)$\\
$C_{14,24,2}$ &$(11,9,13,7, 9, 5)$&$(10, 5,9,7,3, 7)$ &
$C_{30,24,3}$&$(29, 28, 23, 13, 21, 14)$&$(22, 11,  7, 21, 18,  0)$\\
$C_{14,24,3}$ &$( 7,8,11,3,11,12)$&$(13,10,1,6,5, 0)$ &
$C_{30,24,4}$&$(18, 19, 29,  5, 23, 11)$&$(27,  1,  1,  3, 27, 23)$\\
$C_{16,24,1}$ &$(13, 4, 2,11, 4, 4)$&$(15, 8,14, 0, 8,14)$ &
$C_{32,24,1}$&$( 2, 16, 16,  4, 29,  2)$&$(18,  6,  1,  2,  0,  7)$\\
$C_{16,24,2}$ &$( 1, 3,13,15, 5,14)$&$( 5, 9, 3, 7, 9, 5)$ &
$C_{32,24,2}$&$(16,  5, 16, 21, 22,  5)$&$(27, 27,  0, 29,  2,  7)$\\
$C_{16,24,3}$ &$( 8, 9, 8, 3,15, 1)$&$( 6, 0, 2, 3,15, 9)$ &
$C_{32,24,3}$&$(14, 11, 19, 17, 30,  0)$&$(29, 13,  2, 27,  6,  1)$\\
$C_{18,24,1}$&$( 3, 2, 8,15, 8,14)$&$(16, 5, 2,14, 0, 0)$ &
$C_{32,24,4}$&$(15,  3, 31, 23, 29, 27)$&$(27, 10, 13,  3, 19,  3)$\\
$C_{18,24,2}$ &$(15, 3, 8,17,13, 5)$&$(17, 3,17, 1, 3, 5)$ &
$C_{34,24,1}$&$(16, 19, 30, 22, 25, 14)$&$( 0,  3,  4, 32,  2,  0)$\\
$C_{18,24,3}$ &$( 5, 9,13,17, 5, 7)$&$( 4,10, 2, 8, 2, 1)$ &
$C_{34,24,2}$&$(33, 13, 25,  6, 33,  8)$&$(11, 12,  1, 12, 33, 16)$\\
$C_{20,24,1}$&$(17,10,16, 5,10, 4)$&$(12,11,12, 8,14, 8)$ &
$C_{34,24,3}$&$( 4, 25,  3, 27, 15, 13)$&$( 3, 29,  9,  5, 31, 27)$\\
$C_{20,24,2}$&$(13, 3,17,17, 3,17)$&$( 5, 2,15, 5,15, 1)$ &
$C_{34,24,4}$&$(18,  7,  9, 17, 10, 18)$&$(21, 10,  3, 24, 23, 13)$\\
$C_{20,24,3}$&$( 1, 9, 8,17, 2,13)$&$( 5, 6, 1, 4,13, 8)$ &
$C_{36,24,1}$&$(20,  7, 24, 22, 11,  6)$&$( 2,  8,  0,  7,  4, 24)$\\
$C_{22,24,1}$&$(11,  0, 10, 12, 10, 14)$&$(17,  2,  0, 19,  6, 10)$&
$C_{36,24,2}$&$( 2, 11, 34, 35,  6,  1)$&$(11,  2,  1, 25, 23,  8)$\\
$C_{22,24,2}$&$( 7, 21, 11,  5,  2,  3)$&$(21,  7,  1,  7, 15, 13)$&
$C_{36,24,3}$&$(15, 15, 25,  7,  3, 10)$&$(35,  5,  1, 33, 13, 19)$\\
$C_{22,24,3}$&$(16,  9, 14, 11,  9, 21)$&$(16, 17,  4,  5, 18,  5)$&
$C_{36,24,4}$&$(21, 34, 13, 17, 17, 25)$&$(31, 20,  2, 31, 22, 22)$\\
$C_{22,24,4}$&$( 1, 10,  7, 19,  5, 19)$&$(12, 15, 10,  6, 11,  4)$&
$C_{38,24,1}$&$(16, 36, 16, 32, 25, 26)$&$(16, 29,  0, 22, 25, 14)$\\
$C_{24,24,1}$&$( 7,  4,  6, 17,  8, 18)$&$( 9,  8,  4, 10,  4, 12)$&
$C_{38,24,2}$&$(32, 33, 18,  1, 20,  1)$&$(18,  7, 20, 19,  7,  9)$\\
$C_{24,24,2}$&$( 0,  3, 16, 13, 22, 13)$&$(15,  5,  1,  0, 17,  2)$&
$C_{38,24,3}$&$(17, 17, 27, 17, 26,  9)$&$( 2,  2,  3, 34, 22, 13)$\\
$C_{24,24,3}$&$( 9, 11,  5, 23, 11, 19)$&$( 5, 19,  2, 21,  3, 15)$&
$C_{38,24,4}$&$(29,  3,  4, 15, 17, 31)$&$(11,  3,  3, 33, 13, 17)$\\
$C_{24,24,4}$&$( 7,  3, 16, 12, 14,  3)$&$( 2, 11,  0, 19, 23, 23)$&
$C_{40,24,1}$&$(34, 26, 39,  2, 38,  6)$&$(26, 10,  1,  0,  4, 23)$\\
$C_{26,24,1}$&$(17,  2, 14,  8,  8, 10)$&$(10, 15,  0,  0,  1, 10)$&
$C_{40,24,2}$&$(25, 38, 31, 19, 37,  8)$&$( 7, 20,  1,  0,  9, 22)$\\
$C_{26,24,2}$&$(15, 20,  5, 15, 17,  8)$&$(14,  0, 21, 11,  1, 12)$&
$C_{40,24,3}$&$( 5,  1, 32, 15, 30, 35)$&$(24, 25,  3,  3,  6, 28)$\\
$C_{26,24,3}$&$( 1,  7, 14, 23,  5, 21)$&$( 3, 25,  1, 11, 11, 13)$&
$C_{40,24,4}$&$(37, 13, 15, 19, 27, 15)$&$( 3, 15,  5, 31, 31, 10)$\\
\noalign{\hrule height0.8pt}
\end{tabular}
}
\end{table}

\begin{table}[thbp]
\caption{Extremal Type~II $\ZZ_{2k}$-codes of length $24$ $(k=5,6,\ldots,9)$}
\label{Tab:24-xy1}
\centering
\medskip
{\footnotesize
\begin{tabular}{c|c|ll}
\noalign{\hrule height0.8pt}
Codes & $A$ & \multicolumn{1}{c}{$x$}&\multicolumn{1}{c}{$y$} \\
\hline
$C_{10,24,4}$&$M_{10,24,1}$&$(7,1,6,2,4,1,0,3,5,1,7,7)$&$(0,0,0,0,1,5,9,2,4,0,3,2)$\\
$C_{10,24,5}$&$M_{10,24,1}$&$(6,5,6,0,3,6,3,5,4,0,8,2)$&$(0,0,0,4,0,4,4,6,9,5,7,1)$\\
$C_{10,24,6}$&$M_{10,24,1}$&$(1,7,1,1,0,6,0,4,5,7,1,9)$&$(0,0,0,0,3,0,8,1,5,8,4,9)$\\
$C_{10,24,7}$&$M_{10,24,4}$&$(1,9,5,9,8,5,5,3,5,2,2,4)$&$(0,0,0,0,8,2,3,1,2,1,6,1)$\\
$C_{10,24,8}$&$M_{10,24,6}$&$(9,6,0,6,0,4,7,6,7,8,2,7)$&$(0,0,0,5,4,6,7,2,1,3,4,8)$\\
$C_{10,24,9}$&$M_{10,24,7}$&$(5,4,3,8,8,6,7,0,0,4,0,9)$&$(0,0,0,0,0,0,1,6,5,5,8,7)$\\
$C_{12,24,4}$&$M_{12,24,1}$&$(5,7,9,5,6,7,7,1,0,8,1,10)$&$(0,0,0,0,1,5,1,8,10,2,6,9)$\\
$C_{12,24,5}$&$M_{12,24,1}$&$(11,4,10,8,7,8,2,3,0,0,1,10)$&$(0,0,0,0,10,3,11,10,5,11,6,4)$\\
$C_{12,24,6}$&$M_{12,24,2}$&$(8,4,3,9,8,10,0,5,2,8,7,10)$&$(0,0,0,0,9,2,0,0,7,3,8,3)$\\
$C_{12,24,7}$&$M_{12,24,4}$&$(7,2,11,8,2,0,2,10,6,1,5,0)$&$(0,0,0,1,3,10,2,4,6,11,3,4)$\\
$C_{12,24,8}$&$M_{12,24,5}$&$(3,4,5,1,7,8,4,7,3,4,7,9)$&$(0,0,0,3,6,7,7,6,1,8,10,8)$\\
$C_{12,24,9}$&$M_{12,24,7}$&$(5,6,6,10,1,2,5,9,1,5,5,11)$&$(0,0,1,0,6,7,8,7,0,11,6,10)$\\
$C_{14,24,4}$&$M_{14,24,1}$&
$(7,8,6,6,11,2,6,8,2,1,8,5)$&$(0,0,0,0,9,10,13,0,6,7,11,12)$\\
$C_{14,24,5}$&$M_{14,24,1}$&
$(0,9,6,7,12,8,5,7,4,12,0,8)$&$(0,0,0,0,11,10,1,3,2,2,9,10)$\\
$C_{14,24,6}$&$M_{14,24,3}$&
$(4,5,13,0,6,7,1,2,1,3,1,5)$&$(0,0,0,0,7,8,4,13,10,7,7,6)$\\
$C_{14,24,7}$&$M_{14,24,3}$&
$(4,11,12,13,8,13,10,2,1,4,8,0)$&$(0,0,0,0,6,9,11,0,7,6,2,3)$\\
$C_{14,24,8}$&$M_{14,24,6}$&
$(6,13,6,11,7,0,5,6,1,1,3,7)$&$(0,0,0,0,12,5,4,11,11,4,13,2)$\\
$C_{14,24,9}$&$M_{14,24,8}$&
$(6,6,10,1,2,6,8,3,4,5,13,6)$&$(0,0,0,1,3,3,2,7,7,11,11,1)$\\
$C_{16,24,4}$&$M_{16,24,1}$&
$(1,0,7,7,5,9,7,9,2,13,10,8)$&$(0,0,0,0,7,12,6,3,15,10,2,13)$\\
$C_{16,24,5}$&$M_{16,24,1}$&
$(11,14,13,10,5,11,12,5,1,8,9,13)$&$(0,0,0,3,0,7,9,7,14,0,12,12)$\\
$C_{16,24,6}$&$M_{16,24,1}$&
$(4,11,3,10,13,7,0,1,1,2,7,11)$&$(0,0,0,0,10,14,9,2,0,5,5,7)$\\
$C_{16,24,7}$&$M_{16,24,1}$&
$(14,14,15,10,2,7,15,11,2,0,0,8)$&$(0,0,0,0,1,9,1,6,15,6,10,8)$\\
$C_{16,24,8}$&$M_{16,24,3}$&
$(2,6,2,14,13,10,7,10,3,7,4,10)$&$(0,0,0,0,15,5,14,11,8,1,6,2)$\\
$C_{16,24,9}$&$M_{16,24,7}$&
$(10,14,9,14,3,14,7,1,5,13,9,7)$&$(0,0,0,0,9,4,3,9,12,14,1,12)$\\
$C_{18,24,4}$& $M_{18,24,2}$&
$(2,14,14,12,17,5,7,17,1,3,7,3)$&$(0,0,0,0,6,0,3,8,13,3,1,0)$\\
$C_{18,24,5}$& $M_{18,24,3}$&
$(8,4,13,3,13,4,2,10,0,5,2,0)$&$(0,0,0,0,13,13,7,7,11,17,3,3)$\\
$C_{18,24,6}$& $M_{18,24,3}$&
$(3,15,11,15,7,16,1,8,1,14,8,11)$&$(0,0,0,0,6,13,12,9,12,5,11,0)$\\
$C_{18,24,7}$& $M_{18,24,5}$&
$(11,17,9,14,12,3,11,1,0,17,6,3)$&$(0,0,0,0,8,11,0,13,12,1,10,11)$\\
$C_{18,24,8}$& $M_{18,24,5}$&
$(9,16,0,14,10,1,16,10,0,3,3,6)$&$(0,0,0,0,13,5,10,14,5,5,6,6)$\\
$C_{18,24,9}$& $M_{18,24,7}$&
$(4,3,0,9,0,6,0,3,4,2,3,12)$&$(0,0,0,0,4,3,0,1,4,7,11,2)$\\
\noalign{\hrule height0.8pt}
\end{tabular}
}
\end{table}

\begin{table}[thb]
\caption{Extremal Type~II $\ZZ_{2k}$-codes of length $24$ $(k =10,12,\ldots,15)$}
\label{Tab:24-xy2}
\centering
\medskip
{\footnotesize
\begin{tabular}{c|c|ll}
\noalign{\hrule height0.8pt}
Codes & $A$ & \multicolumn{1}{c}{$x$}&\multicolumn{1}{c}{$y$} \\
\hline
$C_{20,24,4}$&$M_{20,24,1}$&
$(19,7,14,7,1,7,4,10,0,3,17,9)$&$(0,0,0,0,14,1,10,11,15,0,14,19)$\\
$C_{20,24,5}$&$M_{20,24,1}$&
$(16,10,16,5,7,10,15,2,0,0,17,16)$&$(0,0,0,0,10,1,1,9,14,18,19,4)$\\
$C_{20,24,6}$&$M_{20,24,1}$&
$(2,1,15,10,17,11,0,18,4,2,10,4)$&$(0,0,0,0,5,15,9,14,2,17,8,14)$\\
$C_{20,24,7}$&$M_{20,24,2}$&
$(11,0,17,19,0,1,0,4,0,4,0,14)$&$(0,0,0,0,14,10,19,0,9,6,5,9)$\\
$C_{20,24,8}$&$M_{20,24,3}$&
$(8,0,12,0,1,17,1,18,2,1,14,4)$&$(0,0,0,0,0,6,10,5,5,10,5,7)$\\
$C_{20,24,9}$&$M_{20,24,4}$&
$(13,6,19,6,7,3,6,4,0,2,0,2)$&$(0,0,0,1,8,14,5,8,9,3,2,14)$\\
$C_{22,24,5}$&$M_{22,24, 1}$&
$( 18, 20, 15, 19, 8, 18, 1, 6, 0, 2, 4, 7)$&
$( 8, 12, 15, 21, 3, 20, 9, 17, 1, 14, 9, 15)$\\
$C_{22,24,6}$&$M_{22,24, 1}$&
$( 20, 10, 4, 2, 16, 3, 0, 19, 2, 12, 11, 9)$&
$( 18, 3, 2, 18, 1, 13, 12, 10, 14, 14, 19, 14)$\\
$C_{22,24,7}$&$M_{22,24, 3}$&
$( 21, 12, 11, 6, 20, 12, 1, 0, 0, 1, 18, 4)$&
$( 13, 11, 12, 6, 0, 10, 7, 9, 6, 4, 8, 8)$\\
$C_{22,24,8}$&$M_{22,24, 3}$&
$( 19, 14, 3, 4, 15, 15, 7, 6, 0, 9, 17, 19)$&
$( 5, 7, 17, 12, 1, 0, 3, 3, 4, 15, 8, 15)$\\
$C_{22,24,9}$&$M_{22,24, 4}$&
$( 7, 18, 20, 4, 17, 17, 8, 15, 0, 8, 20, 16)$&
$( 2, 6, 20, 0, 1, 18, 4, 11, 3, 1, 0, 10)$\\
$C_{24,24,5}$&$M_{24,24,1}$&$(18,21,21,18,23,9,0,0,0,0,10,8)$&$(10,5,1,8,8,20,12,23,4,6,21,10)$\\
$C_{24,24,6}$&$M_{24,24,1}$&$(9,3,23,3,23,14,0,0,1,1,9,10)$&$(0,18,21,4,2,20,14,4,17,7,5,16)$\\
$C_{24,24,7}$&$M_{24,24,2}$&$(13,21,17,12,7,11,0,0,0,5,9,19)$&$(10,5,21,0,19,12,0,16,20,6,3,10)$\\
$C_{24,24,8}$&$M_{24,24,2}$&$(21,16,20,22,9,13,0,0,0,2,2,15)$&$(5,13,8,15,5,14,3,8,7,1,9,10)$\\
$C_{24,24,9}$&$M_{24,24,4}$&$(4,10,6,3,22,3,0,0,0,0,21,21)$&$(6,16,15,6,20,7,20,8,17,0,18,9)$\\
$C_{26,24,5}$&$M_{26,24,1}$&$(2,11,8,18,7,22,0,0,0,0,7,7)$&$(13,24,16,18,2,17,22,12,16,6,2,5,3)$\\
$C_{26,24,6}$&$M_{26,24,2}$&$(9,18,16,17,23,6,0,0,0,0,19,16)$&$(16,20,12,23,21,11,3,23,1,1,20,17)$\\
$C_{26,24,7}$&$M_{26,24,3}$&$(15,9,2,14,11,2,9,18,0,16,24,24)$&$(5,9,9,21,23,13,5,13,9,13,7,25)$\\
$C_{26,24,8}$&$M_{26,24,5}$&$(20,18,16,17,4,3,0,0,0,5,7,20)$&$(21,1,12,4,25,24,2,16,22,8,14,25)$\\
$C_{26,24,9}$&$M_{26,24,6}$&$(20,21,2,15,15,21,0,0,0,24,12,14)$&$(8,7,17,12,13,4,15,14,1,7,15,9)$\\
$C_{28,24,5}$&$M_{28,24,1}$&$(21,21,4,10,2,8,0,0,0,1,21,2)$&$(2,27,26,23,23,12,23,27,11,25,20,5)$\\
$C_{28,24,6}$&$M_{28,24,1}$&$(20,19,24,13,25,12,20,4,0,10,8,13)$&$(13,17,10,15,8,22,20,18,11,18,24,2)$\\
$C_{28,24,7}$&$M_{28,24,2}$&$(16,26,10,6,20,27,0,0,0,3,23,17)$&$(0,13,0,21,22,27,26,3,17,27,25,7)$\\
$C_{28,24,8}$&$M_{28,24,3}$&$(23,20,12,18,14,26,0,0,0,3,3,11)$&$(4,2,27,18,20,22,18,25,14,3,17,4)$\\
$C_{28,24,9}$&$M_{28,24,4}$&$(3,16,16,15,22,7,0,0,0,0,8,13)$&$(16,16,22,22,15,14,5,14,20,1,25,18)$\\
$C_{30,24,5}$&$M_{30,24,1}$&$(29,29,27,4,2,17,8,29,0,3,11,5)$&$(2,12,9,13,26,21,2,7,13,29,19,11)$\\
$C_{30,24,6}$&$M_{30,24,2}$&$(16,18,23,3,2,13,0,0,0,0,5,8)$&$(16,4,14,26,27,24,2,1,23,28,18,17)$\\
$C_{30,24,7}$&$M_{30,24,2}$&$(19,18,21,18,7,22,0,0,0,1,26,10)$&$(14,15,28,28,26,6,9,16,7,4,4,29)$\\
$C_{30,24,8}$&$M_{30,24,2}$&$(29,9,9,4,27,2,3,25,0,11,22,3)$&$(3,10,29,22,19,19,22,13,25,1,6,3)$\\
$C_{30,24,9}$&$M_{30,24,6}$&$(3,1,16,3,29,26,0,0,0,14,6,4)$&$(6,15,25,9,16,24,1,23,21,17,15,26)$\\
\noalign{\hrule height0.8pt}
\end{tabular}
}
\end{table}

\begin{table}[thb]
\caption{Extremal Type~II $\ZZ_{2k}$-codes of length $24$ $(k =16,17,\ldots,20)$}
\label{Tab:24-xy3}
\centering
\medskip
{\footnotesize
\begin{tabular}{c|c|ll}
\noalign{\hrule height0.8pt}
Codes & $A$ & \multicolumn{1}{c}{$x$}&\multicolumn{1}{c}{$y$} \\
\hline
$C_{32,24,5}$&$M_{32,24,1}$&$(18,10,25,10,3,22,23,6,0,0,28,23)$&$(15,8,10,23,30,10,24,4,27,20,29,0)$\\
$C_{32,24,6}$&$M_{32,24,1}$&$(14,4,9,10,22,3,31,18,0,7,14,12)$&$(30,4,29,22,3,31,16,16,31,12,8,10)$\\
$C_{32,24,7}$&$M_{32,24,2}$&$(21,23,19,23,3,1,23,15,0,8,8,0)$&$(9,5,19,4,23,24,29,13,24,9,29,0)$\\
$C_{32,24,8}$&$M_{32,24,2}$&$(27,27,7,25,0,25,28,31,0,0,9,5)$&$(2,13,27,27,14,16,30,6,18,4,12,9)$\\
$C_{32,24,9}$&$M_{32,24,3}$&$(22,24,21,19,31,31,0,4,0,8,22,18)$&$(8,18,29,12,3,10,20,0,14,28,21,5)$\\
$C_{34,24,5}$&$M_{34,24,1}$&$(3,26,20,2,1,9,20,32,0,2,18,33)$&$(14,31,26,5,0,32,21,25,24,10,28,10)$\\
$C_{34,24,6}$&$M_{34,24,1}$&$(3,2,32,30,30,12,5,22,0,11,20,1)$&$(10,10,15,19,2,8,8,26,32,13,20,3)$\\
$C_{34,24,7}$&$M_{34,24,2}$&$(32,25,16,3,16,22,27,29,0,0,12,18)$&$(31,16,17,3,25,17,29,8,30,9,32,13)$\\
$C_{34,24,8}$&$M_{34,24,4}$&$(15,20,3,23,6,23,6,6,0,12,20,28)$&$(6,26,25,22,4,7,28,6,15,14,30,23)$\\
$C_{34,24,9}$&$M_{34,24,4}$&$(13,10,26,7,11,19,6,24,1,1,23,13)$&$(13,13,7,12,31,26,7,31,7,22,13,8)$\\
$C_{36,24,5}$&$M_{36,24,1}$&$(35,16,10,11,0,33,20,23,0,2,2,4)$&$(35,25,34,0,22,11,35,0,35,11,3,19)$\\
$C_{36,24,6}$&$M_{36,24,1}$&$(28,15,24,19,13,16,8,21,0,0,12,2)$&$(3,9,5,33,7,24,23,3,2,21,22,32)$\\
$C_{36,24,7}$&$M_{36,24,1}$&$(34,25,25,32,20,4,27,20,0,0,31,20)$&$(8,8,0,26,31,24,33,4,34,35,11,26)$\\
$C_{36,24,8}$&$M_{36,24,3}$&$(19,29,26,9,30,6,30,18,0,1,2,22)$&$(3,29,31,29,20,3,13,10,28,33,10,31)$\\
$C_{36,24,9}$&$M_{36,24,4}$&$(4,28,15,7,33,26,13,19,1,1,2,21)$&$(33,9,20,4,6,32,17,24,15,12,22,10)$\\
$C_{38,24,5}$&$M_{38,24,1}$&$(18,5,17,23,6,5,29,1,0,0,23,21)$&$(3,33,36,1,15,11,8,12,19,17,10,33)$\\
$C_{38,24,6}$&$M_{38,24,1}$&$(1,22,16,12,19,14,31,19,0,0,36,24)$&$(31,5,17,28,10,21,18,11,35,25,5,20)$\\
$C_{38,24,7}$&$M_{38,24,1}$&$(34,16,34,7,37,2,17,19,0,0,6,32)$&$(21,4,28,35,26,34,9,30,36,11,0,26)$\\
$C_{38,24,8}$&$M_{38,24,3}$&$(2,22,1,1,37,5,18,9,0,1,19,35)$&$(30,11,12,23,11,32,24,17,29,35,19,33)$\\
$C_{38,24,9}$&$M_{38,24,3}$&$(4,28,24,35,37,13,3,34,0,0,28,12)$&$(30,18,10,10,13,32,15,36,7,6,23,32)$\\
$C_{40,24,5}$&$M_{40,24,1}$&$(11,26,7,16,35,33,17,9,0,1,7,2)$&$(24,36,27,6,1,1,34,18,10,7,16,34)$\\
$C_{40,24,6}$&$M_{40,24,1}$&$(12,2,17,30,8,28,22,39,0,0,9,17)$&$(32,38,35,4,28,9,30,39,19,30,12,20)$\\
$C_{40,24,7}$&$M_{40,24,2}$&$(19,20,22,15,5,23,15,9,0,1,10,33)$&$(3,1,29,34,19,21,1,1,28,15,28,6)$\\
$C_{40,24,8}$&$M_{40,24,3}$&$(6,39,12,39,21,27,1,32,1,1,15,34)$&$(35,29,21,13,5,37,28,14,23,1,28,34)$\\
$C_{40,24,9}$&$M_{40,24,3}$&$(23,16,0,28,5,9,22,22,0,2,32,27)$&$(7,29,11,38,30,21,8,12,26,10,22,24)$\\
\noalign{\hrule height0.8pt}
\end{tabular}
}
\end{table}

\begin{table}[thb]
\caption{Extremal Type~II $\ZZ_{2k}$-codes of length $32$ $(k =4,5,\ldots,10)$}
\label{Tab:32-2}
\centering
\medskip
{\small
\begin{tabular}{c|c|ll}
\noalign{\hrule height0.8pt}
Codes & $A$ & \multicolumn{1}{c}{$x$}&\multicolumn{1}{c}{$y$} \\
\hline
$C_{8,32,3}$&$M_{8,32,1}$&
$(0,0,7,6,4,1,1,3,5,0,3,0,0,1,3,2)$&$(0,0,0,0,3,7,2,2,7,1,0,4,2,4,6,2)$\\
$C_{8,32,4}$&$M_{8,32,1}$&
$(5,5,0,5,0,1,6,3,1,3,3,0,1,1,1,1)$&$(0,0,6,3,4,7,1,6,4,3,4,5,5,7,1,4)$\\
$C_{8,32,5}$&$M_{8,32,2}$&
$(6,0,2,4,4,1,3,6,4,5,1,0,0,0,0,4)$&$(0,0,0,0,0,1,3,7,5,1,7,4,0,1,4,5)$\\
$C_{10,32,3}$&$M_{10,32,1}$&
$(4,1,3,6,5,4,2,6,6,3,0,6,1,1,3,5)$&$(0,0,0,0,2,7,8,7,4,1,0,8,0,6,6,9)$\\
$C_{10,32,4}$&$M_{10,32,1}$&
$(3,3,8,8,2,9,0,4,7,3,1,4,0,1,1,4)$&$(5,2,3,6,1,1,2,9,4,1,7,5,2,2,0,0)$\\
$C_{10,32,5}$&$M_{10,32,2}$&
$(3,9,9,8,2,6,4,1,5,8,6,3,0,0,3,5)$&$(0,0,0,1,1,1,8,8,2,5,4,1,3,3,9,2)$\\
$C_{12,32,3}$&$M_{12,32,1}$&
$(6,7,1,0,8,7,4,3,2,8,0,11,0,1,3,3)$&$(0,0,0,0,8,4,1,8,2,8,4,7,1,2,10,5)$\\
$C_{12,32,4}$&$M_{12,32,1}$&
$(10,3,9,4,3,11,6,11,11,2,0,0,0,1,1,2)$&$(2,10,11,10,1,0,0,7,5,7,0,1,0,7,7,10)$\\
$C_{12,32,5}$&$M_{12,32,1}$&
$(1,8,4,3,7,1,2,11,6,0,1,7,0,0,3,0)$&$(0,11,3,5,7,9,6,9,11,3,11,5,0,1,4,5)$\\
$C_{14,32,3}$&$M_{14,32,1}$&
$(0,2,13,11,1,7,4,3,7,3,12,6,0,1,0,6)$&
$(0,11,11,11,11,8,11,13,8,11,9,7,0,11,7,7)$\\
$C_{14,32,4}$&$M_{14,32,1}$&
$(11,12,1,1,3,9,11,7,13,5,12,13,0,5,8,5)$&
$(8,3,13,13,9,13,4,1,5,0,13,3,3,3,7,2)$\\
$C_{14,32,5}$&$M_{14,32,2}$&
$(10,6,7,5,10,2,5,0,5,4,13,3,0,0,1,1)$&
$(1,4,7,7,10,11,5,11,4,0,11,12,0,1,2,10)$\\
$C_{16,32,3}$&$M_{16,32,1}$&
$(15,0,3,3,12,5,11,0,15,5,14,6,0,0,0,3)$&
$(1,2,9,9,13,15,5,12,11,7,3,14,3,2,11,13)$\\
$C_{16,32,4}$&$M_{16,32,1}$&
$(3,11,1,0,12,3,10,10,10,11,6,12,1,1,3,0)$&
$(15,8,14,1,1,13,2,2,2,6,5,11,0,1,3,6)$\\
$C_{16,32,5}$&$M_{16,32,2}$&
$(1,11,15,3,11,15,8,5,10,8,12,0,0,0,2,7)$&
$(2,14,3,10,2,11,11,0,11,13,1,1,0,1,0,12)$\\
$C_{18,32,3}$&$M_{18,32,1}$&
$(9,4,16,10,2,7,4,7,2,3,2,5,1,0,3,5)$&
$(13,0,9,4,17,16,2,0,0,12,8,2,0,14,3,10)$\\
$C_{18,32,4}$&$M_{18,32,1}$&
$(6,12,5,7,5,6,6,7,16,15,5,9,1,0,2,4)$&
$(3,6,6,4,13,12,2,13,7,12,4,1,3,1,4,9)$\\
$C_{18,32,5}$&
$M_{18,32,2}$&
$(15,17,15,12,10,2,5,7,6,6,12,7,0,1,1,2)$&
$(12,16,1,17,13,16,10,14,7,9,4,9,1,10,0,13)$\\
$C_{20,32,3}$&
$M_{20,32,1}$&
$(1,5,8,10,14,18,4,18,17,11,18,16,0,0,2,6)$&
$(6,12,15,14,6,2,5,16,6,18,6,10,3,8,5,18)$\\
$C_{20,32,4}$&
$M_{20,32,1}$&
$(6,6,15,6,6,3,1,18,1,19,11,10,0,0,3,5)$&
$(8,5,5,9,4,18,9,4,0,5,7,15,2,15,14,18)$\\
$C_{20,32,5}$&
$M_{20,32,1}$&
$(5,13,11,19,5,2,18,8,6,9,19,12,0,0,6,3)$&
$(19,14,13,18,2,14,3,2,10,12,12,8,19,16,12,2)$\\
\noalign{\hrule height0.8pt}
\end{tabular}
}
\end{table}
\end{landscape}

\end{document}